\documentclass[envcountsame]{llncs}
\usepackage{graphicx}
\usepackage{color}
\usepackage{amsmath}
\usepackage{amsfonts}
\usepackage{enumerate}
\usepackage{algorithm, caption}
\usepackage[noend]{algpseudocode}
\usepackage[numbers]{natbib}
\usepackage{multirow}
\usepackage{booktabs}
\usepackage{tablefootnote}[2014/01/26]
\usepackage[para]{footmisc}
\newtheorem{fact}[theorem]{Fact}

\begin{document}
\title{Scalable Approximation Algorithm for Graph Summarization}
\author{Maham Anwar Beg\and
        Muhammad Ahmad \and
        Arif Zaman \and
        Imdadullah Khan
}
\institute{Department of Computer Science, School of Science and Engineering, Lahore University of Management Sciences, Pakistan\\
\email{\{14030016,17030056,arifz,imdad.khan\}@lums.edu.pk}
}
\maketitle
\begin{abstract}
Massive sizes of real-world graphs, such as social networks and web graph, impose serious challenges to process and perform analytics on them. These issues can be resolved by working on a small summary of the graph instead . A summary is a compressed version of the graph that removes several details, yet preserves it's essential structure. Generally, some predefined quality measure of the summary is optimized to bound the approximation error incurred by working on the summary instead of the whole graph. All known summarization algorithms are computationally prohibitive and do not scale to large graphs.  
In this paper we present an efficient randomized algorithm to compute graph summaries with the goal to minimize {\em reconstruction error}. We propose a novel weighted sampling scheme to sample vertices for merging that will result in the least reconstruction error. We provide analytical bounds on the running time of the algorithm and prove approximation guarantee for our score computation. Efficiency of our algorithm makes it scalable to very large graphs on which known algorithms cannot be applied. We test our algorithm on several real world graphs to empirically demonstrate the quality of summaries produced and compare to state of the art algorithms. We use the summaries to answer several structural queries about original graph and report their accuracies.
\end{abstract}
\section{Introduction}\label{section:introduction}
Analysis of large graphs is a fundamental task in data mining, with applications in diverse fields such as social networks, e-commerce, sensor networks and bioinformatics. Generally graphs in these domains have very large sizes - millions of nodes and billions of edges are not uncommon. Massive sizes of graphs make processing, storing and performing analytics on them very challenging. These issues can be tackled by working instead on a compact version (summary) of the graph, which removes certain details yet preserves it's essential structure.

Summary of a graph is represented by a `supergraph' with weights both on edges and vertices. Each supernode of the summary, represents a subset of original vertices while it's weight represents the density of subgraph induced by that subset. Weights on edges, represent density of the bipartite graph induced by the two subsets. Quality of a summary is measured by the `reconstruction error', a norm of the difference of actual and reconstructed adjacency matrices. Another parameter adopted to assess summaries is the accuracy in answer to queries about original graph computed from summaries only. 

Note that, since there are exponentially many possible summaries (number of partitions of vertex set), finding the best summary is a computationally challenging task. GraSS \cite{lefevre2010grass} uses an agglomerative approach, where in each iteration a pair of nodes is merged until the desired number of nodes is reached. Since the size of search space at iteration $t$ is $O ({n(t)\choose 2})$, where $n(t)$ is number of supernodes at iteration $t$. GraSS randomly samples $O(n(t))$ pairs and merges the best pair (pair with the least score) among them. With the data structures of GraSS, merging and evaluating score of a pair can be done in $O(\Delta(t))$ (maximum degree of the summary in iteration $t$). This results in the overall worst case complexity of $O(n^2\Delta)$ to compute a summary with $O(n)$ nodes. $S2L$\cite{riondato2014graph} on the other hand uses a clustering technique for Euclidean space by considering each vertex as an N-dimensional vector. The complexity of this algorithm is $O(n^2t)$ to produce a summary of a fixed size $k=O(n)$, where $t$ is the number of iterations before convergence.

In this paper we take the agglomerative approach to compute summary of any desired size. In every iteration a pair is chosen for merging from a randomly chosen sample. We derive a closed form formula for reconstruction error of the graph resulting after merging a pair. Exact computation of this score takes $O(\Delta)$ time but with constant extra space per node, this can be approximated in constant time with bounded error. Furthermore, we define weight of each node that can be updated in constant time and closely estimate the contribution of a node to score of pairs containing it. We select a random sample of pairs by selecting nodes with probability proportional to their weights, resulting in samples of much better quality. We establish that with these weights, logarithmic sized sample yields comparable results. The overall complexity of our algorithm comes down to $O(n(\log n + \Delta))$. Our approach of sampling vertices according to their weights form a dynamic graph (where weights are changing) may be of independent interest. We evaluate our algorithm on several benchmark real world networks and demonstrate that we significantly outperform GraSS\cite{lefevre2010grass} and $S2L$\cite{riondato2014graph} both in terms of running time and quality of summaries. 

The remaining paper is organized as follows. Section \ref{section:related_work} discusses previous work on graph summarization and related problems. In section \ref{section:problem_definition} we formally define the problem with it's background. We present our algorithm along with it's analysis in section \ref{section:algorithm}. In section \ref{section:experiments} we report results of experimental evaluation of our algorithm on several graphs. We also provide comparisons with existing solutions both in terms of runtime and quality.
\section{Related Work}\label{section:related_work}
Graph summarization and compression is a widely studied problem and has applications in diverse domains. There are broadly two types of graph summaries (represented as supergraphs described above): {\em lossless} and {\em lossy}. The original graph can be exactly reconstructed from a lossless summary, hence the goal is to optimize the space complexity of a summary \cite{storer1988data}. The guiding principle here is that of {\em minimum description length} MDL \cite{rissanen1978modeling}. It states that minimum extra information should be kept to describe the summarized data. In lossless summary  \cite{navlakha2008graph,khan2015set}, {\em edge-corrections} are stored along with each supernode and super edges to identify missing edges. \cite{koutra2014vog} stores information about structures like {\em cliques, stars, and chains} formed by subgraphs as lossless summary of a graph.

Lossy compression, on the other hand, compromises some detailed information to reduce the space complexity. There is a trade off  between quality and size of the summary. Quality of a summary is measured by a norm of difference between original adjacency matrix and the adjacency matrix reconstructed from the summary, known as reconstruction error. \cite{lefevre2010grass} adopted an agglomerative approach to greedily merge pairs of nodes to minimize the $l_1$-reconstruction error. Runtime of their algorithm amounts $O(n^3)$ in the worst case. 

In \cite{riondato2014graph}, each node is considered a vector in $\mathbb{R}^n$ (it's row in the adjacency matrix) and point-assignment clustering methods (such as $k$-means) are employed. Each cluster is considered a supernode and the goal is to minimize the $l_2$-reconstruction error. The authors suggest to use dimensionality reduction techniques for points in $\mathbb{R}^n$. This technique does not use any structural information of the graph.
In \cite{zhuang2016data} social contexts and characteristics are used to summarize social networks. Summarization of edge-weighted graphs is studied in \cite{toivonen2011compression}. Graph compression techniques relative to a certain class of queries on labeled graphs is studied in \cite{fan2012query}. \cite{liu2012approximate} uses entropy based unified model to make a homogeneous summary for labeled graphs. Compression of web graphs and social networks is studied in \cite{boldi2004webgraph,adler2001towards} and \cite{chierichetti2009compressing}, respectively.See \cite{liu2016graph} for detailed overview of graph summarization techniques.

A closely related area is that of finding clusters and communities in a graph using  iterative algorithms \cite{newman2004finding}, agglomerative algorithms \cite{clauset2004finding} and spectral techniques \cite{white2005spectral}. Identification of web communities in web graphs using maximum flow/minimum cut problem is discussed in \cite{flake2000efficient}.
\section{Problem Definition}\label{section:problem_definition} Given a graph $G = (V,E)$ on $n$ vertices, let $A$ be it's adjacency matrix. For $k \in \mathbb{Z}$, a summary of $G$, $S=(V_S,E_S)$ is a weighted graph on $k$ vertices. Let $V_S = \{V_1,\ldots,V_k\}$, each $V_i\in V_S$ is referred to a supernode and represents a subset of $V$. More precisely, $V_S$ is a partition of $V$, i.e. $V_i\subset V$ for $1\leq i\leq k$, $V_i\cap V_j =\emptyset$ for $i\neq j$ and $\bigcup_{i=1}^k V_i = V$. Each supernode $V_i$ is associated with two integers $n_i=|V_i|$ and $e_i= |\{(u,v)| u,v \in V_i, (u,v) \in E \}|$. For an edge $(V_i,V_j)\in E_S$ (known as superedge), let $e_{ij}$ be the number of edges in the bipartite subgraph induced between $V_i$ and $V_j$, i.e.	$e_{ij}= |\{(u,v)|u \in V_i, v \in V_j, (u,v) \in E \}|$.
Given a summary $S$, the graph $G$ is approximately reconstructed by the expected adjacency matrix, $\bar{A}$, where $\bar{A}$ is a $n\times n$ matrix with $$\bar{A}(u,v) = \begin{cases} 0 & \text{ if } u = v\\ \frac{e_i}{{n_i\choose 2}} & \text{ if } u,v\in V_i\\ \frac{e_{ij}}{n_in_j} & \text{ if } u\in V_i, v\in V_j \end{cases}$$
The quality of a summary $S$ is assessed by $l_p$-norm of element-wise difference between $\bar{A}$ and $A$.
\begin{figure}
	\centering
	\includegraphics[width=\textwidth]{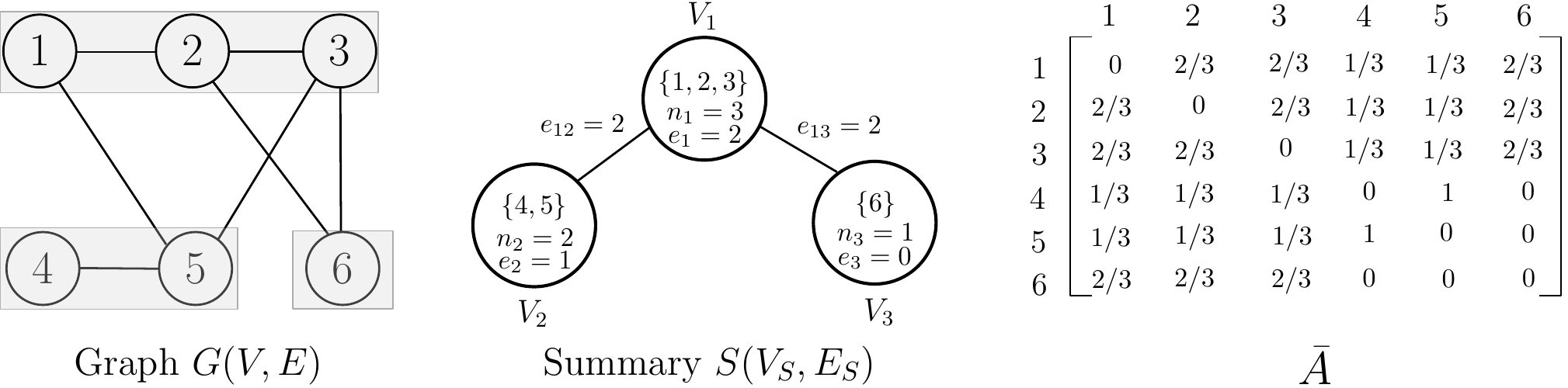}
\end{figure}

\begin{definition}($l_p$-Reconstruction Error ($RE_p$)): The (unnormalized) $l_p$ reconstruction error of a summary $S$ of a graph $G$ is \begin{equation}
		RE_p(G|S) =  RE_p(A | \bar{A}) =  \left(\sum_{i=1}^{| V |} \sum_{j=1}^{| V |} |\bar{A} (i,j) - A(i,j)|^p\right)^{1/p} \label{eq:RE}\end{equation}
\end{definition}
Note that the case $p=1$ considered in \cite{lefevre2010grass} and $p=2$ considered in \cite{riondato2014graph} are closely related to each other. In this paper we use $p=1$ and refer to $RE_1(G|S)$ as $RE(G|S)$. A simple calculation shows that $RE(G|S)$ can be computed in the following closed form.
\begin{equation} \label{eq:RE_closedForm}
	RE(G|S) =  RE(A | \bar{A}) = \sum_{i=1}^{k}4e_i - \frac{4e_i^2}{\binom{n_i}{2}}  + \sum_{i=1}^{k}\sum_{j=1,j\neq i}^{k}  2e_{ij} - \frac{2e_{ij}^2}{n_in_j}
\end{equation}
Formally, we address the following problem. 
\begin{problem}	\label{problem:1}
	Given a graph $G(V,E)$ and a positive integer $k \leq |V|$, find a summary $S$ for $G$ with $k$ super nodes such that $RE(G|S) $ is minimized.
\end{problem}
Another measure to assess quality of a summary $S$ of $G$ is by the accuracy of  answers of queries about structure of $G$ based on $S$ only. In the following we list how certain queries used in the literature are answered from $S$. 

\textbf{Adjacency Queries:} Given two vertices $u,v \in V$, the query whether $(u,v)\in E$ is answered with $\bar{A}(u,v)$. This can either be interpreted as the expected value of an edge being present between $u$ and $v$ or as returning a `yes' answer based on the outcome of a biased coin.

\textbf{Degree Queries:} Given a vertex $v\in V$, the query about degree of $v$ is answered as $\bar{d}(v) =  \sum_{j=1}^{n} \bar{A}(v,j)$.

\textbf{Eigenvector-Centrality Queries:} Eigenvector-centrality of a vertex $v$, $p(v)$ measures the relative importance of $v$ \cite{motwani2010randomized}. For a vertex $v\in V$, this query is answered as $\bar{p}(v)=\frac{\bar{d}(v)}{2| E |}$.

\textbf{Triangle density queries:} Let $t(G)$ be the number of triangles in $G$. $t(G)$ is estimated from $S$ by counting the expected number of triangles within each super node, the expected number of triangles made with one vertex in one supernode and two in another, and that made with one vertex each from three different super nodes. More precisely, this query is answered as follows. Let $\pi_{i} = \frac{e_i}{\binom{n_i}{2}}$ and $\pi_{ij} = \frac{e_{ij}}{n_in_j}$, then $\bar{t}(G)$, the estimate for $t(G)$, is  $$\sum\limits_{\substack{i=1}}^{k} \left[  
\binom{n_i}{3}\pi_{i}^3 + \sum\limits_{\substack{j=i+1}}^{k} \left(\pi_{ij}^2\left[\binom{n_i}{2}n_j\pi_{i}   + \binom{n_j}{2}n_i\pi_{j}  \right] + \sum\limits_{\substack{l=j+1}}^{k}n_in_jn_l\pi_{ij}\pi_{jl}\pi_{il} \right)\right].$$
\section{Algorithm}\label{section:algorithm}
Given a graph $G$ and an integer $k$ our algorithm produces a summary $S$ on $k$ super nodes as follows. Let $S_{t-1}$ be the summary before iteration $t$ with $n(t-1)$ super nodes, i.e. $S_0 = G$, and let $\bar{A}_t$ be the expected adjacency matrix of $S_t$. For $1\leq t \leq n-k$, we select a pair of supernodes $(u,v)$  and merge it to get $S_{t}$. To select an approximately optimal pair thie weight of each node $v$ that closely estimate the contribution of  node to score of pairs $(v,*)$. We randomly sample nodes for each pair with probability proportional to their weights and evaluate score of the pairs. We derive a closed form formula to evaluate score of a pair. Furthermore, in is form these scores can be approximately computed very efficiently. Based on approximate score we select the best pair in the sample and merge it to get $S_{t}$. In what follows, we discuss implementation of each of these subroutines and their analyses. 
\begin{lemma}\label{lem:merge}
	A pair $(u,v)$ of nodes in $S_t$, can be merged to get $S_{t+1}$ in time $O(deg(u) + deg(v))$.  
\end{lemma}
\begin{proof}
	In the adjacency list format, one needs to iterate over neighbors of each $u$ and $v$ and record their information in a new list of the merged node. However, updating the adjacency information at each neighbor of $u$ and $v$ could potentially lead to traversal of all the edges. To this end, as a preprocessing step, for each $(x,y)$, in the adjacency list of $x$ at node $y$, we store a pointer to the corresponding entry in the adjacency list of $y$. With this constant (per edge) extra book keeping we can update the merging information at each neighbor in constant time by traversing just the list of $u$ and $v$. It is easy to see that this preprocessing can be done in time $O(|E|)$ once at the initialization. \qed \end{proof}
The next important step is to determine the {\em goodness} of a pair $(a,b)$. This can be done by temporarily merging $a$ and $b$ and then evaluating \eqref{eq:RE} or \eqref{eq:RE_closedForm} respectively taking $O(n^2)$ and $O(n(t))$. For a pair of nodes $(a,b)$ in $S_{t-1}$, let $S_t^{a,b}$ be the graph obtained after merging $a$ and $b$. We define score of a pair $(a,b)$ to be \begin{align} \label{eq:score} score_t(a,b)&=RE(G|S_{t-1}) - RE(G|S_t^{a,b}) \notag \\
	&= -\frac{4e_a^2}{\binom{n_a}{2}}  - \sum\limits_{\substack{i=1\\i \neq a}}^{n(t)} \frac{4e_{ai}^2}{n_an_i}  +\frac{4e_{ab}^2}{n_an_b}  - \frac{4e_b^2}{\binom{n_b}{2}}-  \sum\limits_{\substack{i=1\\i \neq b}}^{n(t)}\frac{4e_{bi}^2}{n_bn_i} \notag\\ 
	&+ \frac{4\big(e_a+e_b+e_{ab}\big)^2}{\binom{n_a+n_b}{2}} +\frac{4}{\big(n_a+n_b\big)}\sum\limits_{\substack{i=1\\i \neq a,b}}^{n(t)}\Big(\frac{e_{ai}^2}{n_i}+ \frac{e_{bi}^2}{n_i} + \frac{2e_{ai}e_{bi}}{n_i}\Big) 
\end{align}
\begin{fact}
	Since $S_{t-1}$ is fixed, minimizing $RE(G|S_t^{a,b})$ is equivalent to maximizing $score_t(a,b)$. 
\end{fact}
\begin{remark}\label{remark_compute_Score}
	Except for the last summation in \eqref{eq:score} all other terms of $score_t(a,b)$ can be computed in constant time. Since $n_a,n_b,e_a,$ and $e_b$ are already stored at $a$ and $b$, this can be achieved by storing an extra real number $D_a$ at each super node $a$ such that, $D_a = \sum_{\substack{i=1\\i \neq a}}^{n(t)} \frac{e_{ai}^2}{n_i}$. Note that $D_a$ can be updated in constant time after merging of any two vertices $x,y \neq a$, i.e.  after merging $x,y$, while traversing their neighbors for $a$ we subtract $e_{xa}/n_x$ and $e_{ya}/n_y$ from $D_a$ and add back $(e_{x}+e_y)/(n_x+n_y)$ to it. This value can be similarly updated at the merged node too. \end{remark}
The last summation in \eqref{eq:score}, $\sum_{\substack{i=1\\i \neq a,b}}^{n(t)}\frac{2e_{ai}e_{bi}}{n_i}$, in essence is the inner product of two $n(t)$ dimensional vectors ${\cal A}$ and ${\cal B}$, where the $i^{th}$ coordinate of ${\cal A}$ is $\frac{e_{ai}}{\sqrt{n_i}}$ (${\cal B}$ is similarly defined). Storing these vectors will take $O(n(t))$, moreover computing score will take time $O(n(t))$. However, $\left<{\cal A},{\cal B}\right> = {\cal A}\cdot {\cal B}$ can be very closely approximated with a standard application of {\em count-min sketch} \cite{CormodeM05}.

\begin{theorem}\label{thm:CMSketch} (c.f \cite{CormodeM05} Theorem 2) 
	For $0< \epsilon, \delta < 1$, let $\left<\widehat{ {\cal A},{\cal B} }\right>$ be the estimate for $\left<{\cal A},{\cal B}\right>$ using the count-min sketch. Then 
	\begin{itemize}
		\item $\left<\widehat{ {\cal A},{\cal B} }\right> \geq \left<{\cal A},{\cal B}\right>$
		\item $Pr[\left<\widehat{ {\cal A},{\cal B} }\right> < \left<{\cal A},{\cal B}\right> + \epsilon ||{\cal A}||_1||{\cal B}||_1] \geq 1-\delta$.
		${{\cal A},{\cal B}}$
	\end{itemize}
	Furthermore, the space and time complexity of computing $\left<\widehat{ {\cal A},{\cal B} }\right>$ is $O(\frac{1}{\epsilon}\log \frac{1}{\delta})$. While after a merge, the sketch can be updated in time $O(\log\frac{1}{\delta})$.
\end{theorem}

Hence, for a pair of nodes $(a,b)$ in $S_{t-1}$, $score_t(a,b)$ can be closely approximated in constant time. Note that the bounds on time and space complexity, though constants are quite loose in practice.

The next important issue the quadratic size of search space. This is a major hurdle to scalability to large graphs. We define weight of a node $a$ as \begin{equation}
	\label{eq:weight}
	f(a) = -\dfrac{4e_a^2}{\binom{n_a}{2}}  - \sum\limits_{\substack{i=1\\i \neq a}}^{n(t)} \dfrac{4e_{ai}^2}{n_an_i} \hskip.5in w(a) =\begin{cases}
		\dfrac{-1}{f(a)} & \text{ if } f(a)\neq 0\\
		0 & \text{ otherwise}
\end{cases} \end{equation}
We select pairs by sampling nodes according to their weights so as the pairs selected will likely have higher scores. With this weighted sampling a sample of size $O(\log n)$ outperforms a random sample of size $O(n)$. Let $W= \sum_{i=1}^{n(t)} w(i)$ be the sum of weights, we select a vertex $w$ with probability $w(a)/W$. Weighted sampling though can be done in linear time at a given iteration. In our case it is very challenging since the population varies in each iteration; two vertices are merged into one and weights of some nodes also change. To overcome this challenge, we design special data structure ${\cal D}$ that has the following properties.
\begin{claim}\label{claim:dataStructure}
	${\cal D}$ can be implemented as a binary tree such that
	\begin{enumerate}[i.]
		\item it can be initially populated in $O(n)$,
		\item a node can be sampled with probability proportional to it's weight in $O(\log n)$
		\item inserting, deleting or updating a weight in ${\cal D}$ takes time $O(\log n)$. 
	\end{enumerate}  
\end{claim}
\begin{remark}
	We designed this data structure independently, but found out that it has been known to the statistics community since 1980 \cite{wong1980efficient}. We note that this technique could have many applications in sampling from dynamic graphs.
\end{remark}
Algorithm \ref{ourAlgo} is our main summarization algorithm that takes as input $G$, integers $k$ (target summary size), $s$ (sample size), $w$ and $d$ (where $w=\frac{1}{\epsilon}$ and $d = \log \frac{1}{\delta}$ are parameters for count-min sketch). 

\begin{algorithm}[ht]
	\caption{: ScalableSumarization($G=(V,E)$, $k$, $w$, $d$)}
	\label{ourAlgo} 
	\begin{algorithmic}[1]
		\State ${\cal D} \gets \Call{buildSamplingTree}{V,W,1,n}$ \Comment{$W[1\ldots n]$ is initialize as $W[i] = w(v_i)$}
		\While{$G$ has  more than $k$ vertices}
		\State $samplePairs \gets\Call{GetSample}{{\cal D},s}$ \Comment{$s$ calls to Algorithm \ref{alg:sample}}
		\State $scores \gets\Call{ComputeApproxScore}{samplePairs}$ \Comment{Uses \eqref{eq:score} and Theorem \ref{thm:CMSketch}}
		\State  $bestPair \gets\Call{Max}{scores}$
		\State $\Call{Merge}{bestPair}$ \Comment{Lemma \ref{lem:merge}}
		\For{each neighbor $x$ of $u,v \in bestPair$}
		\State $\Call{UpdateWeight}{x,{\cal D}}$
		\EndFor
		\EndWhile
	\end{algorithmic}
\end{algorithm}

For each vertex $a$ we maintain a variable $D_a$ (Remark \ref{remark_compute_Score}). Hence the weight array can be initialized in $O(n)$ time using \eqref{eq:weight}. By Claim \ref{claim:dataStructure}, ${\cal D}$ can be populated in $O(n)$ time. By Claim \ref{claim:dataStructure}, Line 3 takes $O(s\log n)$ time, by Theorem \ref{thm:CMSketch} and \eqref{eq:score} Line 4 takes constant time per pair, and by Lemma \ref{lem:merge} merging can be performed in $O(\Delta)$ time. Since delete and update in ${\cal D}$ takes time $O(\log n)$ and the while loop is executed $n-k+1$ times, total runtime of Algorithm \ref{ourAlgo} is $O((n-k+1)(s\log n + \Delta\log n)$. Generally $k$ is $O(n)$ (typically a fraction of $n$) and in our experiments we take $s$ to be $O(\log n)$ and $O(\log^2 n)$. Furthermore, since many real world graphs are very sparse, ($\Delta$ which is worst case upper bound is constant), we get that overall complexity of our algorithm is $O(n\log^2 n)$ or $O(n\log^3n)$. 

\textbf{Data Structure for sampling:} We implement ${\cal D}$ as a balanced binary tree, where leaf corresponds to (super) node in the graph and stores weight and id of the node. Each internal node stores the sum of values of the two children. The value of the root is equal to $\sum_{i=1}^{n(t)} w(i)$. Furthermore, at each node in the graph we store a pointer to the corresponding leaf. We give pseudocode to construct this tree in Algorithm \ref{alg:buildTree} along with the structure of a tree node. By construction, it is clear that hight of the tree is $\lceil \log n \rceil$ and running time of building the tree and space requirement of ${\cal D}$ is $O(n)$.

\begin{minipage}[t]{.66\textwidth}
	\vskip-.4in
	\begin{algorithm}[H]
		%	\floatname{algorithm}{Procedure}
		\caption{BuildSamplingTree($A$,$W$,$st$,$end$)} \label{alg:buildTree}
		%	\caption*{MyAlgorithm} \label{alg:MyAlgorithm}
		\begin{algorithmic}[1]
			\If{$A[st] = A[end]$}
			\State $leaf \gets \Call{CreateNode}$
			\State $leaf.weight \gets W[st]$
			\State $leaf.vertexID \gets A[st]$ 
			\State \Return $leaf$
			\Else
			\State $mid = \frac{end+st}{2}$
			\State $left \gets \Call{CreateNode}$
			\State $left \gets \Call{BuildSamplingTree}{A,W,st,mid}$
			\State $right \gets \Call{CreateNode}$
			\State $right \gets \Call{BuildSamplingTree}{A,W,mid+1,end}$
			\State $parent \gets \Call{CreateNode}$
			\State $parent.weight \gets left.weight + right.weight$
			\State $left.parent \gets parent$
			\State $right.parent \gets parent$
			\State \Return $parent$
			\EndIf
		\end{algorithmic}
	\end{algorithm}
\end{minipage}
\begin{minipage}[t]{.28\textwidth}
	\vskip-.4in
	\begin{algorithm}[H]
		\renewcommand{\thealgorithm}{}
		\floatname{algorithm}{Structure}
		\caption{TreeNode}
		%	\caption*{MyAlgorithm} \label{alg:MyAlgorithm}
		\begin{algorithmic}
			\State {\bf int } $vertexID$
			\State {\bf double} $weight$
			\State {\bf TreeNode} $*left$
			\State {\bf TreeNode} $*right$
			\State {\bf TreeNode} $*parent$
		\end{algorithmic}
	\end{algorithm}
\end{minipage}

The procedure to sample a vertex with probability proportional to its weight using ${\cal D}$ is given in Algorithm \ref{alg:sample}. This takes as input a uniform random number $r \in [0,\sum_{i=1}^{n(t)} w(i)]$. Since it traverses a single path from root to leaf, the runtime of this algorithm is $O(\log n)$. The update procedure is very similar, whenever weight of a node changes, we start from the corresponding leaf (using the stored pointer to leaf) and change weight of that leaf. Following the parent pointers, we update weights of internal nodes to the new sum of weights of children. Deleting a node is very similar, it amounts to updating weight of the corresponding leaf to $0$. Inserting a node (the super node representing the merged nodes) is achieved by changing the weight of the first empty leaf in ${\cal D}$. A reference to first empty node is maintained as a global variable.
\begin{algorithm}
	\renewcommand{\thealgorithm}{3}
	\caption{:GetLeaf($r$,$node$)}\label{alg:sample}
	\begin{algorithmic}[1]	
		\If{$node.left=$ {\bf NULL} $\wedge  node.right=$ {\bf NULL}}
		\State  \Return $node.vertexID$
		\EndIf	
		\If{$r < node.left.weight$}
		\State \Return $\Call{GetLeaf}{r,node.left}$ 
		\Else
		\State \Return $\Call{GetLeaf}{r-node.weight,node.right}$ 
		\EndIf	
	\end{algorithmic}
\end{algorithm}
\section{Evaluation}\label{section:experiments}
We evaluate performance of our algorithm in terms of runtime, reconstruction error and accuracies of answers to queries on standard benchmark graphs \footnote{\url{http://snap.stanford.edu/}}. We demonstrate that our algorithm substantially outperforms existing solutions, GraSS\cite{lefevre2010grass} and $S2L$\cite{riondato2014graph} in terms of quality while achieving order of magnitude speed-up over them. Our Java Implementation is available at \footnote{\url{https://bitbucket.org/M\_AnwarBeg/scalablesumm/}}. We also report the accuracies in query answered based on summaries only and show that error is very small and we save a lot of time. Errors reported are normalized by $|V|$. All runtimes are in seconds.

\begin{figure}[t]
	\centering
	\includegraphics[scale=.35,page=1]{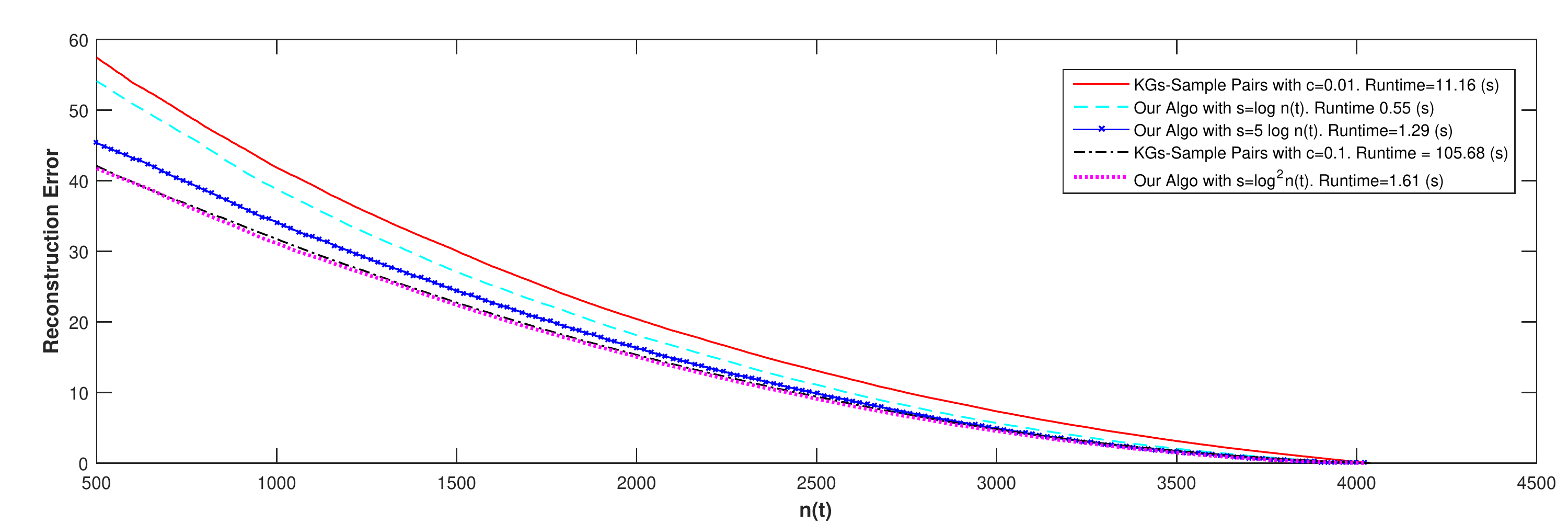} 
	\caption{Comparison between $k$-Gs-SamplePairs on ego-Facebook. Runtimes reported are at k=500.} \label{fig:compgrass}
\end{figure}
From Figure \ref{fig:compgrass}, it is clear that the quality of our summaries compares well with that of $k$-Gs-SamplePairs but with much smaller sample size.  We report results for $s\in \{\log n(t), 5\log n(t), \log^2 n(t)\}$, with exact score computation.  Indeed with sublinear sample size $O(\log n(t))$ and $O(\log^2 n(t))$, our reconstruction error is less than $k$-Gs-SamplePairs with sample size $0.01n(t)$. Although for $n(t)$ significantly smaller than $|V|$, there is a small difference in size of logarithmic and linear sample, but we benefit from our logarithmic sample size for large $n(t)$.

\begin{table}[ht]
	\centering
	\scriptsize
	\renewcommand\tabcolsep{4pt}
	\begin{tabular}{c c c c c c c c c c c c c c }
		\hline		
		&\multicolumn{2}{c}{}&\multicolumn{2}{c}{Error}&\multicolumn{2}{c}{RunTime}&$S2L$\\ 					
		\cmidrule(lr){4-5} 	 \cmidrule(lr){6-7} 		
		Graph&k&w&$RE$&$(l_2)^2$&avg&std&$(l_2)^2$\\ \hline
		\multirow{6}{*}{\shortstack[c]{ego-Facebook\\ $|V|$=4,039 \\ $|E|$=88,234 }}&\multirow{3}{*}{1000}&50&69.98&35.00&0.83&0.01&	\multirow{3}{*}{581}\\ 
		&&100&57.27&28.66&0.94&0.05\\ 
		&&-\tablefootnote{Represents score computation without approximation using Equation \ref{eq:score}.}&38.98&19.51&1.01&0.08&\\ 
		\cmidrule(lr){2-8} 			
		&\multirow{3}{*}{1500}&50&58.17&29.09&0.82&0.05&	\multirow{3}{*}{501}\\ 
		&&100&40.05&20.04&0.83&0.03\\ 
		&&-&27.14&13.58&0.87&0.04\\\hline			
		\multirow{6}{*}{\shortstack[c]{email-Enron\\ $|V|$=36,692 \\$|E|$=183,831 }}&\multirow{3}{*}{10000}&50&6.03&3.01&1.95&0.05&	\multirow{3}{*}{72}\\ 
		&&100&5.84&2.92&2.14&0.02\\ 
		&&-&5.82&2.91&2.21&0.07\\ 
		\cmidrule(lr){2-8} 
		&\multirow{3}{*}{14000}&50&4.25&2.13&1.69&0.05&	\multirow{3}{*}{54}\\ 
		&&100&4.16&2.08&1.77&0.02\\ 
		&&-&4.15&2.08&1.81&0.01\\ \hline			
		\multirow{6}{*}{\shortstack[c]{web-Stanford\\ $|V|$=281,903 \\$|E|$=1,992,636 }}&\multirow{3}{*}{2000}&50&24.11&12.06&65.97&0.02&	\multirow{3}{*}{48}\\ 
		&&100&24.01&12.01&64.41&0.82\\ 
		&&-&24.01&12.05&73.48&0.28\\ 
		\cmidrule(lr){2-8} 
		&\multirow{3}{*}{10000}&50&23.65&11.83&59.87&0.15&	\multirow{3}{*}{38}\\ 
		&&100&23.61&11.81&66.21&0.67\\ 
		&&-&21.13&10.57&59.34&0.28\\ \hline
		\multirow{6}{*}{\shortstack[c]{amazon0601\\ $|V|$=403,394 \\ $|E|$=2,443,408 }}&\multirow{3}{*}{2000}&50&24.11&12.06&147.3&3.72&	\multirow{3}{*}{53}\\ 
		&&100&24.10&12.05&158.03&1.59\\ 
		&&-&24.10&12.05&175.29&5.81\\ 
		\cmidrule(lr){2-8} 
		&\multirow{3}{*}{8000}&50&23.77&11.89&143.83&3.58&	\multirow{3}{*}{51}\\ 
		&&100&23.74&11.87&154.34&1.92\\ 
		&&-&23.73&11.87&171.00&5.63\\ \hline
	\end{tabular}
	\caption{Comparison of $S2L$ with our algorithm using different count min sketch widths. The numbers reported here for $S2L$ are as given by authors in \cite{riondato2014graph}.} \label{tbl:comparisonS2L}
\end{table}

In Table \ref{tbl:comparisonS2L}, we present  reconstruction errors on moderately large sized graphs. Even though $S2L$ is suitable for Euclidean errors, our algorithm still outperforms $S2L$ because by minimizing $RE$ we preserve the original structure of the graph. We use $s=\log n(t)$, $w \in \{50,100\}$ and $d = 2$. We also report results without approximation in score computation. Note that we are able to generate summaries with much smaller runtime on a less powerful machine, (Intel(R) Core i5 with 8.00 GB RAM and 64-bit OS) compared to one reported in \cite{riondato2014graph}. 
\begin{table}[h]
	\centering
	\scriptsize		
	\renewcommand\tabcolsep{5pt}
	\begin{tabular}{c c c c c c c c c c c c c c}\hline		
		&&\multicolumn{3}{c}{as-Skitter     }&\multicolumn{3}{c}{wiki-Talk	}&\multicolumn{3}{c}{com-Youtube	}\\
		&&\multicolumn{3}{c}{$|E|$= 1,696,415   }&\multicolumn{3}{c}{$|E|$ = 2,394,385}&\multicolumn{3}{c}{$|E|$ =1,157,828} \\	
		&&\multicolumn{3}{c}{$|V|$= 11,095,298  }&\multicolumn{3}{c}{$|V|$ = 4,659,565	}&\multicolumn{3}{c}{$|V|$ = 2,987,624	}\\
		\cmidrule(lr){3-5}	\cmidrule(lr){6-8} \cmidrule(lr){9-11}			
		\multicolumn{2}{c}{Parameters}&&\multicolumn{2}{c}{Time}&&\multicolumn{2}{c}{Time}&&\multicolumn{2}{c}{Time}\\
		\cmidrule(lr){1-2} 	\cmidrule(lr){4-5} 	\cmidrule(lr){7-8} \cmidrule(lr){10-11}
		k$\times(10^3)$&w&$RE$&avg&std&$RE$&avg&std&$RE$&avg&std\\
		\multirow{4}{*}{10}&50& 25.42&521.43&12.93& 7.28&311.10&5.68& 9.98&207.38&4.66\\ 
		&100& 24.81&516.03&22.97& 7.07&328.19&2.06& 9.64&222.22&7.17\\ 
		&200& 24.35&559.91&13.07& 6.82&363.37&8.63& 9.30&251.94&8.28\\ 
		&-& 23.81&649.82&20.44& 6.72&319.95&28.91& 9.26&242.58&13.85\\ \hline
		\multirow{4}{*}{50}&50& 23.49&481.40&14.11& 5.77&285.89&5.56& 7.49&184.67&3.86\\ 
		&100& 21.78&480.85&23.54& 5.65&299.98&1.97& 7.13&195.85&6.41\\ 
		&200& 20.90&524.94&12.88& 5.63&329.39&8.95& 7.09&215.87&7.23\\ 
		&-& 20.77&591.35&19.36& 5.62&273.24&24.05& 7.08&199.48&10.85\\ \hline
		\multirow{4}{*}{100}&50& 21.28&436.84&11.52& 5.12&266.15&4.83& 5.90&160.11&3.51\\ 
		&100& 18.90&445.27&23.18& 5.09&276.32&2.77& 5.82&167.67&5.22\\ 
		&200& 18.48&486.88&13.14& 5.08&303.44&8.60& 5.81&183.64&5.81\\ 
		&-& 18.42&535.02&18.90& 5.08&248.73&20.44& 5.81&164.80&9.87\\ \hline
		\multirow{4}{*}{250}
		&50& 15.34&332.27&9.58& 4.23&223.65&3.85& 3.91&103.25&3.97\\ 
		&100& 13.79&350.47&21.77& 4.22&232.39&1.81& 3.90&107.79&1.03\\ 
		&200& 13.68&376.89&11.86& 4.21&256.05&6.74& 3.89&118.73&4.60\\ 
		&-& 13.65&392.58&13.48& 4.21&203.93&18.21& 3.89&98.70&4.89\\ \hline
	\end{tabular}
	\caption{Quality in terms of $RE$ of summary produced by our algorithm on large sized graph.} \label{tbl:resultsLargeGraph}
\end{table}

Table \ref{tbl:resultsLargeGraph} contains quality and runtime for very large graphs, on which none of the previously proposed solutions were applicable. We use $s=\log n(t)$, $w \in \{50,100,200\}$ and $d = 2$. We get some reduction in running times by approximating the score while the quality of summaries remains comparable. Note that large values of $w$ result in increased runtime without any improvement in quality. This is so because  the complexity of exact score evaluation depends on maximum degree, which in real world graphs is small. 
\begin{table}
	\scriptsize
	\begin{minipage}{0.5\linewidth}
		\centering
		\begin{tabular}{c c c c c c c c c  }			\multicolumn{9}{c}{ego-Facebook} \\ \hline
			
			&\multicolumn{2}{c}{}&\multicolumn{2}{c}{}&		\multicolumn{2}{c}{Eigenvector-}&&\\
			
			&\multicolumn{2}{c}{} &	\multicolumn{2}{c}{}&\multicolumn{2}{c}{Centrality}&&\\
			&\multicolumn{2}{c}{}&\multicolumn{2}{c}{Degree}&			\multicolumn{2}{c}{($\times 10^{-5}$)} &\\
			\cmidrule(lr){4-5}\cmidrule(lr){6-7}&\multicolumn{2}{c}{}&\multicolumn{2}{c}{ }&\multicolumn{2}{c}{}&Triangle\\    
			&k&w&avg&stdev&avg&stdev&Density\\\hline
			&\multirow{3}{*}{500}&50&22.50&29.70&6.37&8.42&-0.89\\ 
			&&100&14.95&26.58&4.24&7.53&-0.76\\ 
			&&-&12.01&12.70&3.40&3.60&-0.28\\ 
			\cmidrule(lr){2-8}
			&\multirow{3}{*}{1000}&50&17.66&25.68&5.00&7.28&-0.84\\ 
			&&100&10.67&24.19&3.02&6.85&-0.58\\ 
			&&-&7.62&9.41&2.16&2.67&-0.15\\ %\cmidrule(lr){2-10} 
			\cmidrule(lr){2-8}
			&\multirow{3}{*}{1500}&50&13.22&22.19&3.75&6.29&-0.72\\ 
			&&100&7.04&18.32&2.00&5.19&-0.41\\ 
			&&-&4.50&5.67&1.28&1.61&-0.08\\ \hline
		\end{tabular}
	\end{minipage} 
	\begin{minipage}{0.5\linewidth}
		\scriptsize
		\begin{tabular}{c c c c c c c c c  }

			\multicolumn{9}{c}{email-Enron} \\\hline

			&\multicolumn{2}{c}{}&\multicolumn{2}{c}{}&		\multicolumn{2}{c}{Eigenvector-}&&\\

			&\multicolumn{2}{c}{}&\multicolumn{2}{c}{}&		\multicolumn{2}{c}{Centrality}&&\\
			&\multicolumn{2}{c}{}&\multicolumn{2}{c}{Degree}&			\multicolumn{2}{c}{($\times 10^{-5}$)}&\\
			\cmidrule(lr){4-5}\cmidrule(lr){6-7}
			&\multicolumn{2}{c}{}&\multicolumn{2}{c}{ }&\multicolumn{2}{c}{}&Triangle\\			
			&k&w&avg&stdev&avg&stdev&Density\\\hline			
			&\multirow{4}{*}{4000}&50&2.49&8.60&0.34&1.17&-0.37\\ 
			&&100&1.94&4.26&0.26&0.58&-0.19\\ 
			&&-&1.91&3.09&0.26&0.42&-0.16\\ 
			\cmidrule(lr){2-8}
			&\multirow{4}{*}{6000}&50&1.56&4.70&0.21&0.64&-0.18\\ 
			&&100&1.37&2.33&0.19&0.32&-0.12\\ 
			&&-&1.38&2.24&0.19&0.31&-0.11\\ 
			\cmidrule(lr){2-8}
			&\multirow{4}{*}{8000}&50&1.15&3.10&0.16&0.42&-0.11\\ 
			&&100&1.06&1.77&0.14&0.24&-0.08\\ 
			&&-&1.04&1.52&0.14&0.21&-0.08\\ \hline
		\end{tabular}
	\end{minipage} \caption{Error in queries computed by summaries generated by our algorithm. Absolute average error is reported for degree and centrality query. For triangle density, relative error is reported.}  \label{tbl:queryemail}
\end{table}

In Table \ref{tbl:queryemail}, we tabulate answers to queries that are computed from summaries only. We report mean absolute errors in estimated degrees and eigenvector-centrality scores. For triangle density we report relative error, calculated as $\frac{\bar{t(G)}-t(G)}{t(G)}$. In all cases query answers are very close to the true values, signifying the fact that our summaries do preserve the essential structure of the graph.
\section{Conclusion}\label{section:conclusion}
In this work we devise a sampling based efficient approximation algorithm for graph summarization. We derive a closed form for measuring suitability of a pair of vertex for merging. We approximate this score with theoretical guarantees on error. Another major contribution of this work is the efficient weighted sampling scheme to improve the quality of samples. This enables us to work with substantially smaller sample sizes without compromising summary quality. Our algorithm is scalable to large graphs on which previous algorithms are not applicable. Extensive evaluation on a variety of real world graphs show that our algorithm significantly outperforms existing solutions both in quality and time complexity.
\bibliographystyle{splncs}     
\bibliography{649}  
\end{document}